\documentclass{article}[12pt]
\usepackage{cite}
\usepackage[alphabetic]{amsrefs}
\usepackage{amsthm, amsfonts, amssymb, amsmath}
\usepackage{mathtools}
\usepackage{physics}
\usepackage{tikz}
\usepackage{tkz-graph}
\usepackage[framemethod=tikz]{mdframed}
\usetikzlibrary{decorations.pathreplacing}
\usepackage{varwidth}
\usepackage{xfrac}
\usepackage{enumitem}

\newcommand\comment[1]{}

\newcommand{\prob}[2]{P_{#1}(#2)}
\newcommand{\probblank}[1]{P(#1)}
\newcommand{\probbig}[2]{P_{#1}\big(#2\big)}
\newcommand{\absbig}[1]{\big | #1 \big |}
\newcommand\smallgiven{\:\vert\:}

\newtheorem{Theorem}{Theorem}
\newtheorem{Lemma}{Lemma}
\newtheorem*{Open Question*}{Open Question}
\newtheorem{Definition}{Definition}
\newtheorem{game}{Game}

\usepackage{mathtools}
\DeclarePairedDelimiter{\ceil}{\lceil}{\rceil}

\newcommand{\Prop}{\mathcal{P}}

\DeclareMathOperator{\poly}{poly}

\usepackage[margin=1in]{geometry}

\title{Can graph properties have exponential quantum speedup?}
\author{Andrew M.\ Childs\thanks{Department of Computer Science, Institute for Advanced Computer Studies, and Joint Center for Quantum Information and Computer Science, University of Maryland} \and Daochen Wang\thanks{Applied Mathematics, Statistics, and Scientific Computation, and Joint Center for Quantum Information and Computer Science, University of Maryland}}
\date{}

\begin{document}
\maketitle
\begin{abstract}
Quantum computers can sometimes exponentially outperform classical ones, but only for problems with sufficient structure. While it is well known that query problems with full permutation symmetry can have at most polynomial quantum speedup---even for partial functions---it is unclear how far this condition must be relaxed to enable exponential speedup. In particular, it is natural to ask whether exponential speedup is possible for (partial) \emph{graph properties}, in which the input describes a graph and the output can only depend on its isomorphism class.

We show that the answer to this question depends strongly on the input model. In the adjacency matrix model, we prove that the bounded-error randomized query complexity $R$ of any graph property $\mathcal{P}$ has $R(\mathcal{P}) = O(Q(\mathcal{P})^{6})$, where $Q$ is the bounded-error quantum query complexity. This negatively resolves an open question of Montanaro and de Wolf in the adjacency matrix model. More generally, we prove $R(\mathcal{P}) = O(Q(\mathcal{P})^{3l})$ for any $l$-uniform hypergraph property $\mathcal{P}$ in the adjacency matrix model. 
In direct contrast, in the adjacency list model for bounded-degree graphs, we exhibit a promise problem that shows an exponential separation between the randomized and quantum query complexities.
\end{abstract}

\section{Introduction}

Quantum computers offer the prospect of solving certain problems exponentially faster than is possible classically. The first concrete hint of this possibility came from the model of query complexity, where the input is provided by a black box and the computational cost is quantified as the number of queries to that box. In this model, there is an exponential quantum speedup between deterministic classical and quantum computation \cite{deutsch_jozsa_1992}, and even between bounded-error classical and quantum computation \cite{simon_algorithm_94}. Indeed, these algorithms directly motivated Shor's algorithm for factoring, which replaces the black box with an efficiently computable function to give an (apparent) exponential speedup for an explicit problem \cite{Sho97}.

Since query complexity provides a useful testbed for understanding the potential power of quantum computers, it is natural to explore which problems allow for quantum speedup in this model. In the negative direction, it has been known for over twenty years that for total functions---i.e., query problems that are defined for any possible input string---quantum computers can offer at most a polynomial advantage \cite{beals_buhrman_cleve_mosca_dewolf_polynomial_1998}. More precisely, suppose the goal is to compute some known function
\begin{equation}
    \mathcal{P}\colon S \rightarrow\{0,1\}
\end{equation}
on a black-box input $x \in S \subseteq \Sigma^m$, where $\Sigma$ is a finite set (the \emph{input alphabet}). The domain $S$ of $\mathcal{P}$ is referred to as the \emph{promise} on the input. When $S = \Sigma^m$, we say $\mathcal{P}$ is \emph{total}; otherwise we say it is \emph{partial}. The deterministic, randomized, and quantum query complexities of $\mathcal{P}$ are denoted $D(\mathcal{P}) \ge R(\mathcal{P}) \ge Q(\mathcal{P})$, respectively. Beals et al.\ show that $D(\mathcal{P}) = O(Q(\mathcal{P})^6)$ \cite{beals_buhrman_cleve_mosca_dewolf_polynomial_1998}.

This result establishes that a promise is necessary to achieve superpolynomial quantum speedup. Thus it is natural to ask what kinds of promises can and cannot allow for a significant quantum advantage. In particular, if a query problem is highly symmetric, then superpolynomial quantum speedup remains impossible, even if the problem is partial. Specifically, Aaronson and Ambainis show that if $\Sigma = \{0,1\}$, $S$ is closed under permutations of the input bits, and $\mathcal{P}$ is invariant under those permutations, then $R(\mathcal{P}) = O(Q(\mathcal{P})^2)$ \cite{aaronson_ambainis_structure_14}*{Appendix}. Subsequently, Chailloux showed that $R(\mathcal{P}) = O(Q(\mathcal{P})^3)$ for any $\Sigma$ \cite{chailloux_symmetric_18}. In fact, Ben-David shows that $R(\mathcal{P}) = O(Q(\mathcal{P})^{18})$ even if we only know that $S$ is closed under input permutations \cite{bendavid_structure_promises_16}.

On the other hand, with less than full permutation symmetry, it is unclear when quantum speedups are possible. A natural class of problems with significant symmetry, though much less than full permutation symmetry, is the class of \emph{graph properties}. For such problems, the input describes a graph, and the output depends only on the isomorphism class of that graph. Thus the vertices can be permuted arbitrarily, but such a permutation induces a structured permutation on the edges, about which queries provide information. A graph property can be partial, i.e., there can be a promise that the input graph comes from a restricted family of (isomorphism classes of) graphs. In particular, partial graph properties arise in the setting of \emph{graph property testing}, where the goal is to determine whether a given graph either has a certain property or is far from having that property.

Ambainis, Childs, and Liu studied quantum algorithms for graph property testing, giving polynomial quantum speedups for testing expansion and bipartiteness of bounded-degree graphs in the adjacency list model \cite{ambainis_childs_liu_property_testing_2011}. Furthermore, they showed that at most a polynomial advantage is possible for testing expansion, and asked whether an exponential speedup is ever possible for graph property testing. Montanaro and de Wolf raised this question in a way that can be construed more generally, asking the following:

\begin{Open Question*} {\cite{montanaro_dewolf_property_16}}
Is there any graph property $\mathcal{P}$ which admits an exponential quantum speed-up?
\end{Open Question*}

This question takes different forms depending on the model of access to the graph. Indeed, we show that its answer depends strongly on which of two common input models (formalized in Section~\ref{sec:prelim}) is used: the \emph{adjacency matrix model} (in which the algorithm inputs a pair of vertices and the black box indicates whether they are adjacent) or the \emph{adjacency list model} (in which the algorithm inputs a vertex and the black box returns its neighbours).

When the graph is specified in the adjacency matrix model, we prove in Section~\ref{sec:adjmat} that exponential quantum speedup is impossible. In fact, we prove this not just for graph property testing, but for any partial graph property. Furthermore, we prove this for all $l$-uniform graph properties provided $l$ is constant. Our proof is based on the framework of Chailloux~\cite{chailloux_symmetric_18}, which essentially exploits results of Zhandry~\cites{zhandry_how_construct_quantum_random_12, zhandry_note_collision_set_equality_15}.

In direct contrast, we show that exponential quantum speedup is possible for deciding graph properties in the adjacency list model. Specifically, in Section~\ref{sec:adjlist} we design a promise such that the property $\mathcal{P}_5$ of having a vertex of degree $5$ exhibits an exponential separation between $R(\mathcal{P}_5)$ and $Q(\mathcal{P}_5)$. Our example is based on the glued-trees problem of Childs eta al.\ \cite{childs_2003}.

We conclude in Section~\ref{sec:discussion} with a brief discussion of some open problems.

\section{Preliminaries}\label{sec:prelim}
Throughout the paper, we let $n,d, l\in\mathbb{Z}_{\geq 1}$, $1\leq l\leq n$, and $M\coloneqq \binom{n}{l}$. For $k\in\mathbb{Z}_{\geq1}$, we let $[k] \coloneqq \{1,\dots, k\}$.

A graph property $\mathcal{P}$ is a function from a set of graphs to $\{0,1\}$ that is invariant under graph isomorphisms. For example, ``has a triangle'' is a graph property as two isomorphic graphs either both have a triangle or neither has a triangle.

There are two commonly used models to specify a graph: the adjacency matrix and adjacency list models, introduced below. 

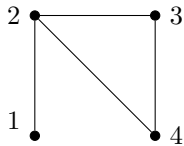
\begin{figure}[ht]\centering
\begin{tikzpicture}[scale=0.8]
      \tikzset{enclosed/.style={draw,circle,inner sep=1pt, minimum size=3.5pt,fill=black}}

      \node[enclosed, label={left, yshift=.2cm: 1}] (1) at (-1,0) {};
      
      \node[enclosed, label={left, yshift=-.0cm: 2}] (2) at (-1, 2) {};
      
      \node[enclosed, label={right, yshift=0cm: 3}] (3) at (1,2) {
      };
      
      \node[enclosed, label={right, yshift=0cm: 4}] (4) at (1,0) {
      };

      \draw (1) -- (2) node[midway, left] (edge1) {};
      
      \draw (2) -- (4) node[midway, above] (edge2) {};
      
      \draw (2) -- (3) node[midway, above] (edge3) {};
      
      \draw (3) -- (4) node[midway, below] (edge4) {};
\end{tikzpicture}
\caption{A graph on $4$ vertices.}
\label{fig:example_graph}
\end{figure}

\subsection{Adjacency matrix model}
An $l$-uniform hypergraph $x$ with vertices $[n]$ is a set $E\subset E_{l}$, where
\begin{equation}\label{def:hypergraph_edges}
     E_l \coloneqq \{\{u_1, u_2, \dots, u_l\} \smallgiven u_i \in [n] \text{ all distinct} \}.
\end{equation}
is the set of all hyperedges. Note that a graph is a $2$-uniform hypergraph and that $|E_l|= \binom{n}{l} = M$.

In the adjacency matrix model, we first fix an identification of $E_l$ with $[M]$. Then we model $x$ by a $M$-bit string $x\in \{0,1\}^M$, or equivalently, function $x\in \{0,1\}^{[M]}$, that indicates the presence $(1)$ or absence $(0)$ of each hyperedge. For example, under the identification
\begin{equation}
    1 \leftrightarrow \{1,2\}, \ 2 \leftrightarrow \{1,3\}, \ 3 \leftrightarrow \{1,4\}, \ 4 \leftrightarrow \{2,3\}, \ 5 \leftrightarrow \{2,4\}, \ 6 \leftrightarrow \{3,4\},
\end{equation}
the graph in Fig.~\ref{fig:example_graph} is specified by $x = 100111$. 

Now, each permutation $\pi \in S_{n}$ of $[n]$ naturally induces a permutation $\Pi \in S_{M}$ of hyperedges $[M]$ by
\begin{equation}\label{eq:graph_property_condition_equiv}
    \Pi(\{u_1, u_2, \dots, u_l\}) = \{\pi(u_1), \pi(u_2), \dots, \pi(u_l)\}.
\end{equation}
Then, saying $\mathcal{P}\colon S\subseteq\{0,1\}^M\rightarrow\{0,1\}$ is a hypergraph property means
\begin{equation}\label{eq:graph_property_invariance}
    x\in S \implies x\circ \Pi \in S \text{ and }
    \mathcal{P}(x) = \mathcal{P}(x \circ \Pi),
\end{equation}
for all induced permutations $\Pi$, where $\circ$ denotes composition of functions.

\subsection{Adjacency list model}
The adjacency list model is an alternative model for specifying a graph, which is particularly well-suited for graphs of bounded-degree. In this model, a graph $x$ on vertices $[n]$ of maximum degree $d$ is modelled by a function
$x\colon [n]\times [d] \rightarrow [n]\cup\{*\}$
with
\begin{equation}
    (u,i) \mapsto \begin{dcases*} v \in [n] & \text{ if} $v$ \text{is the} $i$\text{th neighbour of }$u$, \\ 
    * & \text{ if} $u$ {has fewer than} $i$ neighbours. \end{dcases*}
\end{equation}
Note that $x$ may be non-unique for a given graph due to the choice in the ordering of neighbours. For example, after identifying $x$ with an $n$-by-$d$ grid of entries in $[n]\cup\{*\}$, the graph in Fig.~\ref{fig:example_graph} can be modelled by
\begin{equation}
    x = \begin{bmatrix}
    2 & * & * \\ 
    1 & 3 & 4 \\
    4 & 2 & * \\
    2 & 3 & *
    \end{bmatrix}
\quad\text{or}\quad
    x = \begin{bmatrix}
    2 & * & * \\ 
    4 & 1 & 3 \\
    2 & 4 & * \\
    3 & 2 & *
    \end{bmatrix}
\end{equation}
among other possibilities.

\subsection{Query complexity}
We now briefly discuss query complexity and refer readers to \cite{dewolf_phd} for details. Let $\mathcal{P}$ be any function
\begin{equation}
    \mathcal{P}: S\subseteq \Sigma^m \rightarrow \{0,1\}.
\end{equation}

\begin{Definition}
For $\alpha \in[0.5, 1]$, we say an algorithm $\mathcal{A}$ (randomized or quantum) $\alpha$-estimates $\mathcal{P}$ if, for each $x\in S$, $\mathcal{A}$ outputs $\mathcal{P}(x)$ with probability at least $\alpha$. We say $\mathcal{A}$ estimates $\mathcal{P}$ if $\mathcal{A}$ $(\sfrac{2}{3})$-estimates $\mathcal{P}$. 
\end{Definition}

Then, the (classical) randomized query complexity of $\mathcal{P}$, $R(\mathcal{P})$, is the least $r\in \mathbb{Z}_{\geq 0}$ such that there exists a randomized decision tree that estimates $\mathcal{P}$ and queries at most $r\leq m$ positions of $x$. A randomized decision tree is a decision tree where each node either queries one position of $x$, or randomly draws from some probability distribution to decide the position to query next. 

The quantum query complexity of $\mathcal{P}$, $Q(\mathcal{P})$, is the least $q\in \mathbb{Z}_{\geq 0}$ such that there exists a quantum circuit which estimates $\mathcal{P}$ and queries an $x$-dependent oracle, called $O_x$, at most $q$ times. The oracle $O_x$ is a unitary operator on $\mathbb{C}^m\otimes \mathbb{C}^\Sigma$ defined on basis vectors by
\begin{equation}
    O_x: \ket{i}\ket{y} \mapsto \ket{i}\ket{y \oplus x(i)},
\end{equation}
for $i\in [m], y\in\Sigma$. A quantum algorithm that queries $O_x$ $q$ times means a quantum circuit with $q$ $O_x$ unitaries but any number of unitaries that do not depend on $x$. Such a quantum circuit can be generically written as
\begin{equation}\label{eq:generic_quantum_circuit}
    \mathsf{Q} = U_{q} \, O_x \, U_{q-1} \, \dots U_{1} \, O_x \, U_{0},
\end{equation}
where the unitaries $U_i$ act on space $\mathbb{C}^m \otimes \mathbb{C}^\Sigma \otimes \mathbb{C}^w$ with $\mathbb{C}^w$ being some extra ``work'' register. 

\section{No exponential speedup in the adjacency matrix model}\label{sec:adjmat}

Throughout this section, we fix a hypergraph property $\mathcal{P}$ on vertices $[n]$ and set $q \coloneqq Q(\mathcal{P})$. The aim of this section will be to prove the following.
\begin{Theorem}\label{result:adjacency_matrix}
For any $l$-uniform hypergraph property $\mathcal{P}$, we have
\begin{equation}
    R(\mathcal{P}) = O(Q(\mathcal{P})^{3l}).
\end{equation}
In particular, by setting $l=2$, we have $R(\mathcal{P}) = O(Q(\mathcal{P})^6)$ for graph properties.
\end{Theorem}

In the following, we set $\Sigma = \{0,1\}$, as appropriate for hypergraphs, but the same proof works essentially without modification for any $\Sigma$. 

Our proof strategy closely follows \cite{chailloux_symmetric_18}. We construct a randomized decision tree that estimates $\mathcal{P}$ using $O(q^{3l})$ queries from a quantum circuit that estimates $\mathcal{P}$ using $q$ queries.

By definition of $q$, there exists a $q$-query quantum circuit $\mathsf{Q}'_x$ of the form Eq.~\eqref{eq:generic_quantum_circuit} which estimates $\mathcal{P}$. By repeating $\mathsf{Q}'_x$ three times, and outputting a bit according to majority vote, the probability of success can be boosted to at least $3\cdot\left(\sfrac{2}{3}\right)^2 \cdot\left(\sfrac{1}{3}\right) +  \left(\sfrac{2}{3}\right)^3 = \sfrac{20}{27}$. Therefore, there exists a $3q$-query quantum circuit 
\begin{equation}
    \mathsf{Q}_x = U_{3q} \, O_x \, U_{3q-1} \, \dots U_1 \, O_x \, U_0
\end{equation}
which $(\sfrac{20}{27})$-estimates $\mathcal{P}$. Note that $x\in \{0,1\}^M$ describes the input graph and $O_{x}$ acts on $\mathbb{C}^M \otimes \mathbb{C}^2$.

Let $N\coloneqq \sum_{i=1}^l\binom{n}{i}$. Recall the notation from Eq.~\eqref{def:hypergraph_edges} that $E_i$ is the set of all hyperedges involving $i$ distinct vertices. We identify $[N]$ with the set $\bigcup_{i=1}^lE_{l-(i-1)}$ in the order the union is written. For example, $[M]\subset [N]$ is identified with $E_l$. The reason for defining $N$ and identifying $[N]$ this way shall become clear later, in the second remark following our definition of $D_r$. Suffice it to say now that we shall need to consider functions that map hyperedges connecting $l$ vertices to ones connecting between $1$ and $l$ vertices.

Now let $F$ be a function from $[M]$ to $[N]$. We define the oracle $O_{x\circ F}$, acting on $\mathbb{C}^M \otimes \mathbb{C}^N \otimes \mathbb{C}^2$, as follows.

\begin{enumerate}
    \item Define oracle $\tilde{O}_x$ on $\mathbb{C}^N\otimes \mathbb{C}^2$ by
    \begin{equation}\label{eq:stronger_oracle}
    \tilde{O}_x: \ket{i}\ket{y} \mapsto 
    \begin{dcases*} \ket{i} \ket{y\oplus x(i)} &\text{ if } \hspace{0.25em} $1 \leq i\leq M$,\\ \ket{i}\ket{y} &\text{ if }  $M < i \leq N$.\end{dcases*}
\end{equation}

    \item Define unitary $O_F$ on $\mathbb{C}^M \otimes \mathbb{C}^N$ by
    \begin{equation}
        O_F\colon \ket{i}\ket{j} \mapsto \ket{i}\ket{ j + F(i) \mod{N}}.
    \end{equation}
    
    \item Define oracle $O_{x\circ F}$ on  $\mathbb{C}^M \otimes \mathbb{C}^N \otimes \mathbb{C}^2$ by 
    \begin{equation}
        O_{x\circ F} \coloneqq (O_F^{\dagger} \otimes \mathbb{I}_2)(\mathbb{I}_{M} \otimes \tilde{O}_x)(O_F \otimes \mathbb{I}_2).
    \end{equation}
    
    Then, for $i\in [M]$ and $b\in \Sigma$, we have
    \begin{equation}
        O_{x\circ F}: \ket{i}\ket{0_{\mathbb{C}^N}}\ket{b} \mapsto \ket{i}\ket{0_{\mathbb{C}^N}}\ket{b'},
    \end{equation}
    where
    \begin{equation}\label{eq:x_circ_F_output}
        \ket{b'} \coloneqq \begin{dcases*}  \ket{b\oplus (x\circ F)(i)} &\text{ if } \hspace{.5em}$1 \leq F(i)\leq M$,\\
        \ket{b} &\text{ if }  $M < F(i) \leq N$.
    \end{dcases*}
    \end{equation}
\end{enumerate}

Recall that $O_x$ acts on $\mathbb{C}^M\otimes \mathbb{C}^2$ and $O_{x\circ F}$ acts on  $\mathbb{C}^M \otimes \mathbb{C}^N \otimes \mathbb{C}^2$. Therefore, we may replace each $O_x$ appearing in $\textsf{Q}_x$ by $O_{x\circ F}$ in the natural way that matches the $\mathbb{C}^M\otimes \mathbb{C}^2$ register and leaves an additional $\mathbb{C}^N$ register. We call the resulting circuit $\mathsf{Q}_{x}(F)$. Note that $\mathsf{Q}_{x}(F)$ has $3q$  $O_F$ gates and $3q$ $O_F^{\dagger}$ gates giving a total of $6q$ gates that are each either $O_F$ or $O_F^{\dagger}$.

If $F$ is a permutation of $[M]$, i.e., $F$ bijects its domain $[M]$ with $[M] \subset [N]$, then
$\mathsf{Q}_{x}(F)$ is essentially the same as $\mathsf{Q}_{x\circ{F}}$, albeit with an additional $\mathbb{C}^N$ register. If, in addition, $F$ is a permutation of $[M]$ induced by a permutation of vertices $[n]$, then $\mathcal{P}(x) = \mathcal{P}(x\circ F)$ as $\mathcal{P}$ is a graph property. Therefore, we have
\begin{equation}\label{eq:graph_permuted_circuit_correctness_bound}
    \probblank{\mathsf{Q}_{x}(F) \text{ outputs } \mathcal{P}(x)} = \probblank{Q_{x\circ F} \text{ outputs } \mathcal{P}(x\circ F)} \geq \frac{20}{27}
\end{equation}
as $\mathsf{Q}_x$ $(\sfrac{20}{27})$-estimates $\mathcal{P}$.

The core argument of our proof is that $\mathsf{Q}_{x}(F)$ can behave similarly to $\mathsf{Q}_{x}$ even when $F$ has a limited range. This argument uses the following Theorem~\ref{theorem:zhandry_main} on a suitable family of distributions on functions $F\colon [M] \rightarrow [N]$.

\begin{Theorem}[{\cite{zhandry_how_construct_quantum_random_12}*{Theorem~7.3}}]\label{theorem:zhandry_main}
Let $\mathfrak{D}_r$ be a family of distributions on functions $F\colon [M] \rightarrow [N]$, indexed by $r\in \mathbb{Z}_{\geq 1}\cup \{\infty\}$. Suppose there is an integer $d$ such that the following holds. Fix $q\in \mathbb{Z}_{\geq 0}$, and then fix $2q$ pairs $(d_i,e_i)\in [M]\times [N]$. Suppose there exists a polynomial $p:\mathbb{R}\rightarrow \mathbb{R}$ of degree at most $d$ such that
\begin{equation}
    p(1/r) = \probbig{F \sim \mathfrak{D}_r}{F(d_i) = e_i \textup{ for all } i\in[2q]}.
\end{equation}
for all $r\in \mathbb{Z}_{\geq 1}\cup \{\infty\}$.
Then, for any $\{0,1\}$-output quantum circuit $\mathsf{Q(F)}$ making $q$ quantum queries, each to either\footnote{Strictly speaking, in the original statement of Theorem~\ref{theorem:zhandry_main}, queries are made only to $O_F$. However, no adjustment to the resulting Eq.~\eqref{eq:output_distribution_closeness} is needed when queries can also be made to $O_F^{\dagger}$. This can be easily seen by examining the proof of \cite{zhandry_how_construct_quantum_random_12}*{Theorem~7.1}, given in \cite{zhandry_identity_based_encryption_12}*{Appendix~B.1}, and noting that $O_{F}^{\dagger}$ is simply the unitary on $\mathbb{C}^M\otimes \mathbb{C}^N$ acting by $\ket{i}\ket{j} \mapsto \ket{i}\ket{j - F(i) \mod{N}}$.} $O_F$ or $O_F^{\dagger}$, we have
\begin{equation}\label{eq:output_distribution_closeness}
    \sum_{z\in\{0,1\}} \absbig{\probbig{F\sim \mathfrak{D}_{r}}{\mathsf{Q}(F) \textup{ outputs } z}  - \prob{F\sim \mathfrak{D}_{\infty}}{\mathsf{Q}(F) \textup{ outputs } z}} \leq \frac{\pi^2 d^3}{3 \, r}.
\end{equation}
\end{Theorem}

Note that the probabilities appearing in Eq.~\eqref{eq:output_distribution_closeness} take their natural meaning:
\begin{equation}~\label{eq:composed_distribution}
    \prob{F\sim \mathfrak{D}_r}{\mathsf{Q}(F) \textup{ outputs } z} \coloneqq  \sum_{F} \prob{\,}{\textsf{Q}(F) \text{ outputs } z} \cdot \prob{F \sim \mathfrak{D}_r}{F}.
\end{equation}

To utilise Theorem~\ref{theorem:zhandry_main}, we define, for each $r\in \mathbb{Z}_{\geq 0}$, a distribution $D_r$ on functions $f\colon [n] \rightarrow [n]$ and $F\colon [M] \rightarrow [N]$, where sampling is obtained by the following procedure.

\begin{enumerate}
    \item Draw a random function $g\colon [n]\rightarrow [r]$.
    \item Let $S=\{g(x): x\in [n]\}$. That is, $S$ is the range of $g$.
    \item Draw a random injective function $h\colon S\rightarrow [n]$ (these  functions exist since $|S|\leq n$).
    \item Output $f=h \circ g$.
    \item Output $F\colon [M] \rightarrow [N]$, defined by $F(\{u_1,\dots,u_l\}) \mapsto \{f(u_1),\dots, f(u_l)\}$ for all $\{u,v\}\in [M]$.
\end{enumerate}

For $r = \infty$, we further define $D_{\infty}$ by sampling $f$ as a random permutation on $[n]$ and $F$ as the permutation on $[M]$ induced by $f$. Because $\mathcal{P}$ is a graph property, any such $F$ must satisfy $\mathcal{P}(x) = \mathcal{P}(x \circ F)$, cf. Eq.~\eqref{eq:graph_property_invariance}. Intuitively, $D_\infty$ can be thought of as the limit of $D_r$, with $r\in \mathbb{Z}_{\geq 0}$, as $r\rightarrow \infty$.

We make two remarks. First, the distribution $D_r$ on $f$ is exactly the same as that defined in~\cite{zhandry_note_collision_set_equality_15}*{Sec.~3.1}. Second, in Step 5, because $f$ may map multiple inputs to the same output, sets in the image of $F$ may have fewer than $l$ elements. In fact, they may be any of the $N$ sets in $\bigcup_{i=1}^lE_{l-(i-1)}$ which recall we have identified with $[N]$.

Now, the following Lemma~\ref{lemma:probability_degree_bound} allows us to apply Theorem~\ref{theorem:zhandry_main} to the distribution $D_r$ on $F$. Our Lemma~\ref{lemma:probability_degree_bound} can be deduced from and compared with~\cite{zhandry_note_collision_set_equality_15}*{Lemma 1}.

\begin{Lemma}\label{lemma:probability_degree_bound}
Fix $k\in \mathbb{Z}_{\geq 1}$, and then fix $k$ pairs $(d_i, e_i) \in [M] \times [N]$. Then, there exists a polynomial $p:\mathbb{R}\rightarrow \mathbb{R}$ of degree at most $kl-1$ such that
\begin{equation}
    p(1/r) = \probbig{F\sim D_r}{F(d_i)=e_i \textup{ for all } i\in[k]}
\end{equation}
for all $r \in \mathbb{Z}_{\geq 1}\cup\{\infty\}$.
\end{Lemma}
\begin{proof}
For each $i$, write $d_i = \{u_{1}^{(i)},\dots u_{l}^{(i)}\}$ and $e_i= \{v_{1}^{(i)}, \dots v_{l_i}^{(i)}\}$, where $1\leq l_i \leq l$. By definition, we have
\begin{equation}
    F(d_i)=e_i \iff \{f(u_{1}^{(i)}), \dots, f(u_{l}^{(i)})\} = \{v_{1}^{(i)}, \dots, v_{l_i}^{(i)}\}.
\end{equation}

Write $\{u_j\}_{j=1}^{a} = \bigcup_id_i$ and $\{v_j\}_{j=1}^b = \bigcup_ie_i$, where $a\leq kl$ and $b\leq \sum_{i=1}^k l_i$. Then, the event $\{F(d_i)=e_i$ for all $i\in[k]\}$ can be expressed as a disjoint union of events of the form
\begin{equation}\label{eq:shattered}
    \{f(u_1) = v_{j_1}, \dots, f(u_a) = v_{j_a}\}.
\end{equation}
Therefore, $P_{F\sim D_r}(F(d_i)=e_i$ for all $i\in[k])$ is a sum of probabilities of the form
\begin{equation}\label{eq:shattered_probabilities}
    P_{f\sim D_r}(f(u_1) = v_{j_1}, \dots, f(u_a) = v_{j_a}).
\end{equation}
But \cite{zhandry_note_collision_set_equality_15}*{Lemma 1} says that the probabilities in Eq.~\eqref{eq:shattered_probabilities} can be represented by a polynomial in $1/r$ of degree at most $a-1\leq kl-1$. Hence the Lemma follows.
\end{proof}

Now, setting $k$ to $12q$ in Lemma~\ref{lemma:probability_degree_bound}, we see that the hypothesis of Theorem~\ref{theorem:zhandry_main}, with $q$ set to $6q$, holds for $d = 6ql-1$. Therefore, its conclusion, Eq.~\eqref{eq:output_distribution_closeness}, holds for $Q_{x}(F)$, which recall uses $6q$ queries each to either $O_F$ or $O_F^\dagger$, and all $r$. By setting $r$ equal to
\begin{equation}\label{eq:draw_distribution_s}
    s \coloneqq  \ceil[\Bigg]{\frac{\pi^2 (12ql-1)^3}{3} \cdot \frac{27}{2}}
\end{equation}
in Eq.~\eqref{eq:output_distribution_closeness}, we deduce
\begin{equation}~\label{eq:output_distribution_closeness_applied}
    \sum_{z\in\{0,1\}} \absbig{\prob{F\sim D_s}{\mathsf{Q}_x(F) \textup{ outputs } z}  - \prob{F\sim D_{\infty}}{\mathsf{Q}_x(F) \textup{ outputs } z}} \leq \frac{2}{27}.
\end{equation}

Having defined $s$, we can now describe a randomized decision tree $\mathsf{R}$ that estimates $\mathcal{P}$ as follows.

\begin{mdframed}
$\mathsf{R}$. Given an input graph $x\in \{0,1\}^M$, do the following.
\begin{enumerate}
    \item Draw a random function $F\colon [M]\rightarrow[N]$ according to $D_s$.
    \item Query bits $x(i)$ for those $i\leq M$ in the image of $F$, i.e., $i\in I\coloneqq \textup{Im}(F)\cap[M]$.
    \item Output $z \in \{0,1\}$ according to the output distribution of quantum circuit $\mathsf{Q}_{x}(F)$.
\end{enumerate}
\end{mdframed}

First, by recalling from Eq.~\eqref{eq:x_circ_F_output} the action of $O_{x\circ{F}}$, we see that the output distribution of $\mathsf{Q}_{x}(F)$ only depends on $F$ and the values taken by $x(i)$ for $i\in I$. This output distribution can be \textit{pre-computed} prior to executing $\mathsf{R}$, and $\mathsf{R}$ simply draws from it in Step~3.

Second, let us prove the correctness of $\mathsf{R}$.
\begin{Lemma}
$\mathsf{R}$ estimates $\mathcal{P}$.
\end{Lemma}

\begin{proof}
Let $x\in S\subset \{0,1\}^M$ be an input graph, and write $z \coloneqq \mathcal{P}(x)$.

Since any $F$ drawn from $D_{\infty}$ is a permutation of $[M]$ induced by a permutation of vertices $[n]$, Eq.~\eqref{eq:graph_permuted_circuit_correctness_bound} holds, i.e., $P(\mathsf{Q}_{x}(F) \text{ outputs } z) \geq \frac{20}{27}$. So, by Eq.~\eqref{eq:composed_distribution}, we also have $\prob{F\sim D_{\infty}}{\mathsf{Q}_{x}(F) \text{ outputs } z} \geq \frac{20}{27}$. So
\begin{equation}\label{eq:randomize_estimates}
\begin{aligned}
   \probblank{\mathsf{R} \text{ outputs } z \text{ on input } x} &= \prob{F\sim D_s}{\mathsf{Q}_{x}(F) \text{ outputs } z}\\
   &\geq \prob{F\sim D_{\infty}}{\mathsf{Q}_{x}(F) \text{ outputs } z} - \frac{2}{27} \geq \frac{20}{27} - \frac{2}{27} = \frac{2}{3},
\end{aligned}
\end{equation}
where the first equality follows from the definition of $\mathsf{R}$ and the first inequality follows from Eq.~\eqref{eq:output_distribution_closeness_applied}. 

But $x\in S$ was arbitrary. Therefore, the last bound of Eq.~\eqref{eq:randomize_estimates} says that $\mathsf{R}$ estimates $\mathcal{P}$.
\end{proof}

Now, the query complexity, $\text{cost}(\mathsf{R})$, of $\mathsf{R}$ is the size of the set $I = \textup{Im}(F) \cap [M]$, which is at most $\binom{s}{l}$.
Then, by substituting in the expression for $s$ from Eq.~\eqref{eq:draw_distribution_s}, we deduce
\begin{equation}
    R(\mathcal{P}) \leq \text{cost}(\mathsf{R}) \leq \frac{(43l)^{3l}}{l!}\cdot q^{3l} + (\text{lower order terms in } q) = O(Q(\mathcal{P})^{3l}),
\end{equation}
because $q = Q(\mathcal{P})$ by definition. Therefore, Result~\ref{result:adjacency_matrix} is proved.

\section{Example of exponential speedup in the adjacency list model}\label{sec:adjlist}

We show that an exponential quantum query speedup exists in the adjacency list model by presenting an explicit example. Our example is based on the glued-trees problem of~\cite{childs_2003} but with a modification so that the answer is invariant under vertex relabellings.

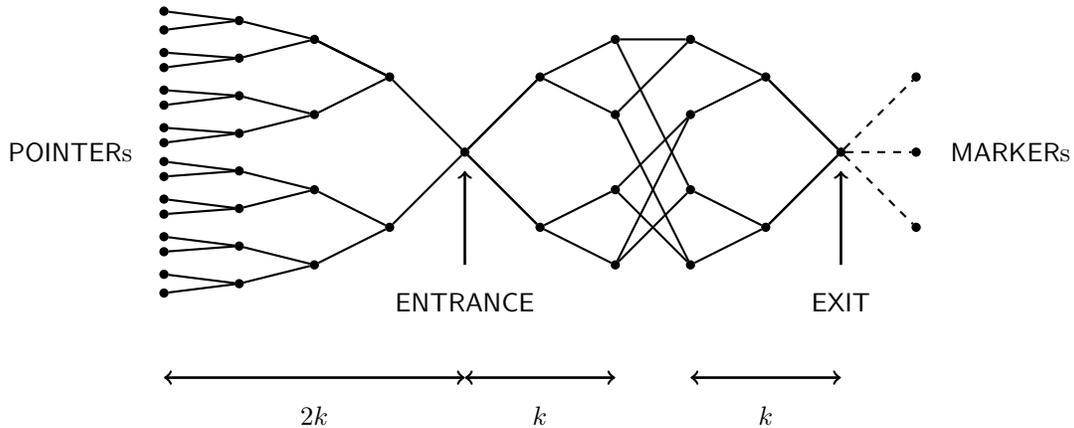
\begin{figure}[ht]
  \centering
  \begin{tikzpicture}[scale=1]
    \SetGraphUnit{1.5}
    \GraphInit[vstyle=Simple]
    \tikzset{VertexStyle/.style = {shape = circle,fill = black,minimum size = 3.5pt,inner sep=1pt}}
    
    \Vertex[x=-1,y=1]{T11}
    \Vertex[x=-1,y=-1]{T12}
    \Vertex[x=-2,y=1.5]{T21}
    \Vertex[x=-2,y=0.5]{T22}
    \Vertex[x=-2,y=-0.5]{T23}
    \Vertex[x=-2,y=-1.5]{T24}
    \Vertex[x=-3,y=1.75]{T31}
    \Vertex[x=-3,y=1.25]{T32}
    \Vertex[x=-3,y=0.75]{T33}
    \Vertex[x=-3,y=0.25]{T34}
    \Vertex[x=-3,y=-0.25]{T35}
    \Vertex[x=-3,y=-0.75]{T36}
    \Vertex[x=-3,y=-1.25]{T37}
    \Vertex[x=-3,y=-1.75]{T38}
    
    \Vertex[x=-4,y=1.875]{T41}
    \Vertex[x=-4,y=1.625]{T42}
    \Vertex[x=-4,y=1.325]{T43}
    \Vertex[x=-4,y=1.125]{T44}
    \Vertex[x=-4,y=0.825]{T45}
    \Vertex[x=-4,y=0.625]{T46}
    \Vertex[x=-4,y=0.325]{T47}
    \Vertex[x=-4,y=0.125]{T48}
    
    \Vertex[x=-4,y=-0.125]{T49}
    \Vertex[x=-4,y=-0.325]{T410}
    \Vertex[x=-4,y=-0.625]{T411}
    \Vertex[x=-4,y=-0.825]{T412}
    \Vertex[x=-4,y=-1.125]{T413}
    \Vertex[x=-4,y=-1.325]{T414}
    \Vertex[x=-4,y=-1.625]{T415}
    \Vertex[x=-4,y=-1.875]{T416}
    
    \Vertex[x=0,y=0]{O}
    
    \Vertex[x=1,y=1]{M11}
    \Vertex[x=1,y=-1]{M12}
    \Vertex[x=2,y=1.5]{M21} 
    \Vertex[x=2,y=0.5]{M22}
    \Vertex[x=2,y=-0.5]{M23}
    \Vertex[x=2,y=-1.5]{M24}
    
    \Vertex[x=3,y=1.5]{H11}
    \Vertex[x=3,y=0.5]{H12}
    \Vertex[x=3,y=-0.5]{H13}
    \Vertex[x=3,y=-1.5]{H14}
    \Vertex[x=4,y=1]{H21}
    \Vertex[x=4,y=-1]{H22}
    \Vertex[x=5,y=0]{E}
    
    \Vertex[x=6,y=1]{Mark1}
    \Vertex[x=6,y=0]{Mark2}
    \Vertex[x=6,y=-1]{Mark3}
    
    \Edges(O,M11,M21,H13,H22,E)
    \Edges(O,M11,M22,H11,H21,E)
    \Edges(O,M12,M23,H14,H22,E)
    \Edges(O,M12,M24,H12,H21,E)
    
    \Edges(M21, H11)
    \Edges(M22, H14)
    \Edges(M23, H12)
    \Edges(M24, H13)
    
    \Edges(T11, O)
    \Edges(T12, O)
    \Edges(T11, T21)
    \Edges(T11, T22)
    \Edges(T12, T23)
    \Edges(T12, T24)
    \Edges(T11, T21)
    \Edges(T11, T21)
    \Edges(T21, T31)
    \Edges(T21, T32)
    \Edges(T22, T33)
    \Edges(T22, T34)
    \Edges(T23, T35)
    \Edges(T23, T36)
    \Edges(T24, T37)
    \Edges(T24, T38)
    \Edges(T31, T41)
    \Edges(T31, T42)
    \Edges(T32, T43)
    \Edges(T32, T44)
    \Edges(T33, T45)
    \Edges(T33, T46)
    \Edges(T34, T47)
    \Edges(T34, T48)
    \Edges(T35, T49)
    \Edges(T35, T410)
    \Edges(T36, T411)
    \Edges(T36, T412)
    \Edges(T37, T413)
    \Edges(T37, T414)
    \Edges(T38, T415)
    \Edges(T38, T416)
    
    \Edges[style=dashed](E, Mark1)
    \Edges[style=dashed](E, Mark2)
    \Edges[style=dashed](E, Mark3)

    \node[text width=3cm, align=center] at (-5.25,0) {\textsf{POINTER}s};
    \node[text width=3cm, align=center] at (5,-2) {\textsf{EXIT}};
    \node[text width=3cm, align=center] at (0,-2) {\textsf{ENTRANCE}};
    \node[text width=3cm, align=center] at (7.25,0) {\textsf{MARKER}s};
    
    \draw [->, line width=1pt] (0,-1.5) -- (0,-.25);
    \draw [->, line width=1pt] (5,-1.5) -- (5,-.25);

    \draw[<->, line width=1pt] (-4,-3) -- (0,-3);
    \node[text width=3cm, align=center] at (-2,-3.5) {$2k$};
    \draw[<->, line width=1pt] (0,-3) -- (2,-3);
    \node[text width=3cm, align=center] at (1,-3.5) {$k$};
    \draw[<->, line width=1pt] (3,-3) -- (5,-3);
    \node[text width=3cm, align=center] at (4,-3.5) {$k$};
  \end{tikzpicture}
  \caption{An illustration of a modified glued-trees graph in the case $k=2$. This graph has a vertex of degree $5$, i.e., has graph property $\mathcal{P}_5$, when the three \textsf{MARKER}s are connected to \textsf{EXIT}. It does not have $\mathcal{P}_5$ when these \textsf{MARKER}s are isolated.}
  \label{fig:modified_glued_tree}
\end{figure}

Let $k \in \mathbb{Z}_{\geq 1}$. A glued trees graph of depth $2k+1$ is a graph consisting of two binary trees of equal depth $k$ joined together at their leaves by any cycle that alternates between the two trees. In Fig.~\ref{fig:modified_glued_tree}, the graph between \textsf{ENTRANCE} and \textsf{EXIT} (inclusive) is a glued tree with $k=2$ of depth $5$. Such graphs have $2(2^{k+1}-1)$ vertices.

In the original glued-trees problem, we are given an oracle that provides a black-box description of a graph that is the union of a glued-trees graph and a (much larger) number of isolated vertices (equivalently, there are many labels that do not refer to any vertex), as well as the label of \textsf{ENTRANCE}. We are required to output the label of \textsf{EXIT}. This is not a graph property because the answer depends on the labelling of the vertices by definition. Moreover, it does not allow an \textit{unconditional} comparison between $R$ and $Q$ since we are advised with the label of a particular vertex $\textsf{ENTRANCE}$. We now describe how we can overcome these two problems.

Consider the set of all glued-trees of depth $k$. To each of them, we append a binary tree of depth $2k$ to the left of \textsf{ENTRANCE}. This forms a new set which we call the set of modified glued trees. We illustrate a modified glued tree in Fig.~\ref{fig:modified_glued_tree}. We call any vertex on the extreme left a \textsf{POINTER}.

Let $A$ be the set of modified glued trees with three extra isolated vertices appended. Let $B$ be the set of modified glued trees with three extra vertices appended that are connected to \textsf{EXIT}. In both cases, we call the appended vertices \textsf{MARKER}s.

We define our promise set $S$ to be $A\cup B$ and ensure that $S$ contains all isomorphic graphs, i.e., all vertex relabellings. All graphs in $S$ have degree at most $d=5$. The graph property $\mathcal{P}_5$ that we consider is whether there exists a vertex of degree $5$. Of course, $\mathcal{P}_5(A) = \{0\}$ and $\mathcal{P}_5(B) = \{1\}$. 

\begin{Theorem}\label{thm:classical_lower}
Graph property $\mathcal{P}_5$, under the promise $S=A\cup B$, has $R(\mathcal{P}_5) = 2^{\Omega(k)}$.
\end{Theorem}

As in the lower bound proof of \cite{childs_2003}, we prove this by reductions between games. First, observe that our problem is essentially equivalent to the following:

\begin{game}\label{game:original}
Given an oracle that provides a black-box description of a graph from $A \cup B$, which upon querying by the label of a vertex, returns as output the labels of all adjacent vertices. The algorithm wins as soon as it queries by the label of \textsf{EXIT}, i.e., the oracle returns either two $2$ or $5$ neighbours.
\end{game}

\begin{Lemma}
Suppose there is an algorithm $\mathcal{A}_{\mathcal{P}_5}$, using at most $T \leq 2^k$ queries, that correctly decides $\mathcal{P}_5$ for each graph in $A\cup B$ with probability at least $P$. Then, there is an algorithm $\mathcal{A}_A$, using at most $T$ queries, that wins Game A for each graph in $A \cup B$ with probability at least $P - O(T/2^{2k})$.
\end{Lemma}

\begin{proof}
First, note that the oracle given in Game A is more powerful than the usual adjacency list oracle in that it outputs \textit{all} adjacent vertex labels. Therefore, one query of $\mathcal{A}_{\mathcal{P}_5}$ can be simulated by at most one query of $\mathcal{A}_A$. We let $\mathcal{A}_A$ simulate $\mathcal{A}_{\mathcal{P}_5}$.

In the case the input graph does not have $\mathcal{P}_5$, $\mathcal{A}_{\mathcal{P}_5}$ can only correctly decide $\mathcal{P}_5$ if it queries by the label of either $\textsf{EXIT}$ or a $\textsf{MARKER}$. In the former case, $\mathcal{A}_{\mathcal{P}_5}$ also wins. In the latter case, we may assume that $\mathcal{A}_A$ loses and that $\mathcal{A}_{\mathcal{P}_5}$ had not queried by the label of \textsf{EXIT}. But then the latter case occurs with probability at most
\begin{equation}
    T\cdot \frac{3}{2^{2k}-5\cdot 2^k} = O(T/2^{2k})
\end{equation}
by the union bound. Of course, the bound can be tightened by increasing the first term in the denominator to equal the total number of vertices in the modified glued-trees graph; but we used $2^{2k}$, i.e., the number of \textsf{POINTERS}, as it suffices to prove the Lemma. The ``$5$'' in the denominator is the maximum degree.

In the case the input graph is from set $B$, $\mathcal{A}_{\mathcal{P}_5}$ must query $\textsf{EXIT}$, because in this case a $\textsf{MARKER}$ is indistinguishable from a $\textsf{POINTER}$ as both have degree $1$. Therefore, $\mathcal{A}_A$ also wins.
\end{proof}

Now consider:

\begin{game}\label{game:more_info_constrained}
Given the same oracle as in Game~\ref{game:original}, the label of \textsf{ENTRANCE} and its four neighbours, and information about which two of the four are in the direction of \textsf{EXIT}. The algorithm wins as soon as it queries by the label of \textsf{EXIT}.
\end{game}

Since we are provided with more information at the outset in Game B than in Game A, it is clear that Game B is no harder to win than Game A. Now consider:

\begin{game}\label{game:hybrid}
Given the same oracle as in Game~\ref{game:original}, the label of \textsf{ENTRANCE} and its four neighbours, and information about which two of the four are in the direction of \textsf{EXIT}. At each step, the algorithm can only query by the label of a vertex to the right of \textsf{ENTRANCE} that had been given or previously returned by the oracle. The algorithm wins as soon as it queries by the label of \textsf{EXIT}.
\end{game}

Game C is not much harder to win than Game B. More precisely:
\begin{Lemma}
Suppose there is an algorithm $\mathcal{A}_B$, using at most $T = 2^{o(k)}\leq 2^k$ queries, that wins Game B for each graph in $A \cup B$ with probability at least $P_B$. Then, there is an algorithm $\mathcal{A}_C$, using at most $T$ queries, that wins Game C for each graph in $A \cup B$ with probability at least $P_B - O(T/2^k)$.
\end{Lemma}

\begin{proof}
Suppose we are given an algorithm $\mathcal{A}_{B}$ that wins Game B with probability at least $P_B$ after $T= 2^{o(k)} \leq 2^k$ queries. We construct an algorithm $\mathcal{A}_C$ for Game C that simulates $\mathcal{A}_B$.

At each step, $\mathcal{A}_B$ can perform one of the following Actions:
\begin{enumerate}
    \item query by the label of a vertex to the right of \textsf{ENTRANCE} that was initially given or previously returned by the oracle, or
    \item query by the label of a vertex that does not satisfy the above condition.
\end{enumerate}

We let $\mathcal{A}_C$ simulate $\mathcal{A}_B$ by the following Actions:
\begin{enumerate}[label=\Roman*.]
    \item when $\mathcal{A}_B$ performs Action 1, $\mathcal{A}_C$ also performs Action 1, or
    \item when $\mathcal{A}_B$ performs Action 2, querying by a label $l$, $\mathcal{A}_C$ first randomly selects a labelling of a depth $2k$ binary tree that is consistent with the labels it has seen so far. Then $\mathcal{A}_C$ imagines that the oracle returns the labels of the neighbours of $l$ according to that labelling.
\end{enumerate}
Note that $\mathcal{A}_C$ does not actually query the oracle in Action II.

Action II can be ill-defined in the event ($E_2$) that the randomly selected labelling does not contain label $l$, in which case we assume that $\mathcal{A}_C$ loses. We also assume that $\mathcal{A}_C$ loses in the event ($E_1$) that, when performing Action I, either of the two returned labels had already appeared during some previous instance of Action II.

If neither event $E_1$ nor $E_2$ occurs at any step of the simulation $\mathcal{A}_C$, then $\mathcal{A}_C$ wins with probability at least $P_B$. But at each step of the simulation, we have
\begin{align}
    p(E_1) &= O \left(2\cdot \frac{5\cdot2^k}{2^{2k}-5\cdot 2^k}\right) = O(2^{-k})
\end{align}
since $p(E_1)$ is upper bounded by the probability that a random set $\mathcal{L}$, with at most $5\cdot 2^k$ labels, contains either of two particular labels within a set of at least $2^{2k}-5\cdot 2^k$ un-queried labels. (Note that the probability distribution of $\mathcal{L}$ is uniform over all subsets of the un-queried labels of size $|\mathcal{L}|$.) Similarly,
\begin{align}
    p(E_2) &= O\left(\frac{2(2^{k+1}-1)}{2^{2k}-5\cdot 2^k}\right) = O(2^{-k})
\end{align}
since $p(E_2)$ is upper bounded by the probability that label $l$ is among labels to the right of $\textsf{ENTRANCE}$.

Therefore, the probability that either $E_1$ or $E_2$ occurs at any step of the simulation is $O(T/2^k)$ by the union bound, and the lemma follows.
\end{proof}

Now, Game C is essentially the same as the following Game D:

\begin{game}[Game 2 of \cite{childs_2003}, STOC version]\label{game:gluedtrees}
Given an oracle that provides a black-box description of a glued trees graph, which upon querying by the label of a vertex, returns as output the labels of all adjacent vertices. Given also the label of the \textsf{ENTRANCE} of this glued trees graph. At each step, the algorithm can only query by the label of \textsf{ENTRANCE} or the label of a vertex that had been previously returned by the oracle. The algorithm wins as soon as  the oracle the label of \textsf{EXIT}.
\end{game}

\begin{Lemma}
Suppose there is an algorithm $\mathcal{A}_C$, using at most $T$ queries, that wins Game C for each graph in $A \cup B$ with probability at least $P_C$. Then there is an algorithm $\mathcal{A}_D$, using at most $T$ queries, that wins Game D for each glued-trees graph with probability at least $P_C$.
\end{Lemma}
\begin{proof}
We may simply let an algorithm $\mathcal{A}_D$ simulate $\mathcal{A}_C$ because the queries made by $\mathcal{A}_C$ are all allowed in Game D, by the design of Game C. 
\end{proof}

Finally, we conclude that the original problem of deciding $\mathcal{P}_5$ requires exponentially many queries due to the exponential lower bound on Game~\ref{game:gluedtrees} \cite{childs_2003}. This establishes Theorem~\ref{thm:classical_lower}.

On the other hand, we show that this problem can be solved efficiently by a quantum computer.

\begin{Theorem}\label{thm:quantum_upper}
Graph property $\mathcal{P}_5$, under the promise $S=A\cup B$, has $Q(\mathcal{P}_5) = \poly(k)$.
\end{Theorem}

\begin{proof}
We describe a two-stage quantum algorithm that estimates $\Prop_5$ using $\poly(k)$ queries.

\begin{enumerate}[labelwidth=4.5em,leftmargin=!]
    \item[Stage 1.] Keep querying a vertex uniformly at random until we hit a \textsf{POINTER} (distinguished by having degree $1$). This occurs with high probability after a constant number of queries because the probability of querying a \textsf{POINTER} is
\begin{equation}
    \frac{2^{2k}}{2\cdot2^{2k}+2(2^{k+1}-2)}\approx \frac{1}{2}
\end{equation}
for large $k$. Then, we perform a classical random walk from the \textsf{POINTER} (never walking backwards to the previously queried vertex). If after $k$ steps we reach another \textsf{POINTER}, we know that we made the wrong turn at the $(k/2)$-th step ($k$ must be even). So we simply proceed by taking the correct turn at the  $(k/2)$-th step. Continuing similarly, we can reach \textsf{ENTRANCE} (distinguished by having degree $4$) after $O(k^2)$ queries~\cite{childs_2003}. Upon reaching \textsf{ENTRANCE} from one direction, we walk $2k$ steps in each of the three other directions (again never walking backwards). One of these directions will lead to another \textsf{POINTER} and we can eliminate that direction as a direction to reach $\textsf{EXIT}$. This concludes the first stage.
    \item[Stage 2.] After the first stage, the problem becomes essentially the same as the original glued trees problem. Therefore, in the second stage, we run the quantum walk algorithm of \cite{childs_2003}.
    This algorithm is able to query \textsf{EXIT} with high probability after $\poly(k)$ queries. Note that \textsf{EXIT} is distinguished by having degree $2$ or $5$ when the input graph is in set $A$ or $B$, respectively.
\end{enumerate}
Since this algorithm finds the \textsf{EXIT} using only $\poly(k)$ queries, the result follows.
\end{proof}

\section{Discussion}\label{sec:discussion}

We have shown that graph properties do not admit an exponential quantum query speedup in the adjacency matrix model, but that there is a graph property with exponential quantum speedup in the adjacency list model. 

We emphasize that this work leaves open the question of whether there can be an exponential separation for graph \textit{property testing} in the adjacency list model. In property testing, the set of ``no'' inputs must include all those that are $\epsilon$-far away from the ``yes'' inputs. Our example in Section~\ref{sec:adjlist} does not satisfy this condition.

Another question is whether the framework of \cite{chailloux_symmetric_18} can be used to address $\mathcal{P}$ with other types of symmetries, or more precisely automorphism groups
\begin{equation}
    \text{Aut}(\mathcal{P})\coloneqq \{\sigma\in S_m \ | \ \mathcal{P}(x) = \mathcal{P}(x\circ\sigma) \text{ for all } x\in S\}.
\end{equation}

Note: as we were completing this manuscript, we learned of \cite{shalev}, which proves a similar result to that of Sec.~\ref{sec:adjmat} also via \cite{chailloux_symmetric_18}.

\section*{Acknowledgments}
We thank Carl Miller for many helpful discussions.
We acknowledge support from the Army Research Office (grant W911NF-20-1-0015); the Department of Energy, Office of Science, Office of Advanced Scientific Computing Research, Quantum Algorithms Teams and Accelerated Research in Quantum Computing programs; and the National Science Foundation (grant CCF-1813814).

\newpage
\bibliographystyle{unsrt}
\bibliography{gp_references.bib}

\begin{bibdiv}
\begin{biblist}

\bib{aaronson_ambainis_structure_14}{article}{
      author={Aaronson, Scott},
      author={Ambainis, Andris},
       title={{The Need for Structure in Quantum Speedups}},
        date={2014},
     journal={Theory of Computing},
      volume={10},
      number={6},
       pages={133\ndash 166},
         url={http://www.theoryofcomputing.org/articles/v010a006},
}

\bib{ambainis_childs_liu_property_testing_2011}{inproceedings}{
      author={Ambainis, Andris},
      author={Childs, Andrew~M.},
      author={Liu, Yi-Kai},
       title={{Quantum Property Testing for Bounded-degree Graphs}},
        date={2011},
   booktitle={{Proceedings of the 14th International Workshop and 15th
  International Conference on Approximation, Randomization, and Combinatorial
  Optimization: Algorithms and Techniques}},
      series={APPROX'11/RANDOM'11},
   publisher={Springer-Verlag},
     address={Berlin, Heidelberg},
       pages={365\ndash 376},
         url={http://dl.acm.org/citation.cfm?id=2033252.2033285},
}

\bib{beals_buhrman_cleve_mosca_dewolf_polynomial_1998}{inproceedings}{
      author={{Beals}, R.},
      author={{Buhrman}, H.},
      author={{Cleve}, R.},
      author={{Mosca}, M.},
      author={{de Wolf}, R.},
       title={{Quantum lower bounds by polynomials}},
        date={1998},
   booktitle={{Proceedings of the 39th Annual Symposium on Foundations of
  Computer Science}},
       pages={352\ndash 361},
}

\bib{bendavid_structure_promises_16}{inproceedings}{
      author={Ben-David, Shalev},
       title={{The Structure of Promises in Quantum Speedups}},
        date={2016},
   booktitle={{11th Conference on the Theory of Quantum Computation,
  Communication and Cryptography, TQC 2016, September 27-29, 2016, Berlin,
  Germany}},
       pages={7:1\ndash 7:14},
         url={https://doi.org/10.4230/LIPIcs.TQC.2016.7},
}

\bib{shalev}{misc}{
      author={Ben-David, Shalev},
      author={Podder, Supartha},
       title={How symmetric is too symmetric for large quantum speedups?},
        date={2020},
        note={arXiv:2001.09642},
}

\bib{childs_2003}{inproceedings}{
      author={Childs, Andrew~M.},
      author={Cleve, Richard},
      author={Deotto, Enrico},
      author={Farhi, Edward},
      author={Gutmann, Sam},
      author={Spielman, Daniel~A.},
       title={{Exponential Algorithmic Speedup by a Quantum Walk}},
        date={2003},
   booktitle={{Proceedings of the Thirty-fifth Annual ACM Symposium on Theory
  of Computing}},
      series={STOC '03},
   publisher={ACM},
     address={New York, NY, USA},
       pages={59\ndash 68},
         url={http://doi.acm.org/10.1145/780542.780552},
}

\bib{chailloux_symmetric_18}{inproceedings}{
      author={Chailloux, Andr{\'e}},
       title={{A Note on the Quantum Query Complexity of Permutation Symmetric
  Functions}},
        date={2018},
   booktitle={{Proceedings of the 10th Innovations in Theoretical Computer
  Science Conference}},
      editor={Blum, Avrim},
      series={Leibniz International Proceedings in Informatics (LIPIcs)},
      volume={124},
   publisher={Schloss Dagstuhl--Leibniz-Zentrum fuer Informatik},
     address={Dagstuhl, Germany},
       pages={19:1\ndash 19:7},
         url={http://drops.dagstuhl.de/opus/volltexte/2018/10112},
}

\bib{deutsch_jozsa_1992}{article}{
      author={{Deutsch}, David},
      author={{Jozsa}, Richard},
       title={{Rapid Solution of Problems by Quantum Computation}},
        date={1992},
     journal={Proceedings of the Royal Society of London Series A},
      volume={439},
      number={1907},
       pages={553\ndash 558},
}

\bib{dewolf_phd}{thesis}{
      author={de~Wolf, Ronald},
       title={{Quantum Computing and Communication Complexity}},
        type={Ph.D. Thesis},
        date={2001},
}

\bib{montanaro_dewolf_property_16}{book}{
      author={Montanaro, Ashley},
      author={de~Wolf, Ronald},
       title={{A Survey of Quantum Property Testing}},
      series={Graduate Surveys},
   publisher={Theory of Computing Library},
        date={2016},
      number={7},
         url={http://www.theoryofcomputing.org/library.html},
}

\bib{Sho97}{article}{
      author={Shor, Peter~W.},
       title={Polynomial-time algorithms for prime factorization and discrete
  logarithms on a quantum computer},
        date={1997},
     journal={SIAM Journal on Computing},
      volume={26},
      number={5},
       pages={1484\ndash 1509},
      eprint={quant-ph/9508027},
        note={Preliminary version in FOCS 1994},
}

\bib{simon_algorithm_94}{inproceedings}{
      author={Simon, D.~R.},
       title={{On the Power of Quantum Computation}},
        date={1994},
   booktitle={{Proceedings of the 35th Annual Symposium on Foundations of
  Computer Science}},
      series={SFCS '94},
   publisher={IEEE Computer Society},
     address={Washington, DC, USA},
       pages={116\ndash 123},
         url={https://doi.org/10.1109/SFCS.1994.365701},
}

\bib{zhandry_how_construct_quantum_random_12}{inproceedings}{
      author={Zhandry, Mark},
       title={{How to Construct Quantum Random Functions}},
        date={2012},
   booktitle={{Proceedings of the 53rd Annual Symposium on Foundations of
  Computer Science}},
      series={FOCS '12},
   publisher={IEEE Computer Society},
     address={Washington, DC, USA},
       pages={679\ndash 687},
         url={https://doi.org/10.1109/FOCS.2012.37},
}

\bib{zhandry_identity_based_encryption_12}{inproceedings}{
      author={Zhandry, Mark},
       title={Secure identity-based encryption in the quantum random oracle
  model},
        date={2012},
   booktitle={{Proceedings of CRYPTO}},
}

\bib{zhandry_note_collision_set_equality_15}{article}{
      author={Zhandry, Mark},
       title={{A Note on the Quantum Collision and Set Equality Problems}},
        date={2015-05},
        ISSN={1533-7146},
     journal={Quantum Info. Comput.},
      volume={15},
      number={7-8},
       pages={557\ndash 567},
         url={http://dl.acm.org/citation.cfm?id=2871411.2871413},
}

\end{biblist}
\end{bibdiv}

\end{document}